\documentclass[a4paper,UKenglish,cleveref, autoref, thm-restate]{lipics-v2021}
\nolinenumbers 

\usepackage{amsmath,amssymb,amsfonts}
\usepackage{array,multicol,multirow}
\usepackage{enumerate,verbatim}
\usepackage{graphicx,listings,tikz}
\usetikzlibrary{calc,positioning,automata,fit,arrows.meta}
\tikzset{arrows={[scale=1]}}
\tikzset{every edge/.style={draw,->,>=Latex,auto}}
\DeclareSymbolFont{rsfscript}{OMS}{rsfs}{m}{n}
\DeclareSymbolFontAlphabet{\mathrsfs}{rsfscript}
\newcommand{\Cerny}{\v{C}ern{\'y} }
\DeclareMathOperator{\IdentityRel}{\vartriangle}
\DeclareMathOperator{\TotalRel}{\triangledown}

\title{Primitive Automata that are Synchronizing}

\author{Igor Rystsov}{National Technical University of Ukraine, Kiev, Ukraine }{haryst49@gmail.com}{https://orcid.org/0000-0001-6002-0503}{}
\author{Marek Szyku{\l}a}{University of Wroc{\l}aw, Faculty of Mathematics and Computer Science, Wroc{\l}aw, Poland}{msz@cs.uni.wroc.pl}{https://orcid.org/0000-0001-5349-468X}{}

\authorrunning{I. Rystsov and M. Szyku{\l}a} 

\Copyright{Igor Rystsov and Marek Szyku{\l}a} 



\category{} 

\relatedversion{} 

\funding{This work was supported in part by the National Science Centre, Poland under project number 2021/41/B/ST6/03691.}

\acknowledgements{}

\EventEditors{John Q. Open and Joan R. Access}
\EventNoEds{2}
\EventLongTitle{42nd Conference on Very Important Topics (CVIT 2016)}
\EventShortTitle{CVIT 2016}
\EventAcronym{CVIT}
\EventYear{2016}
\EventDate{December 24--27, 2016}
\EventLocation{Little Whinging, United Kingdom}
\EventLogo{}
\SeriesVolume{42}
\ArticleNo{23}
\begin{document}
\maketitle
\begin{abstract}
A deterministic finite (semi)automaton is primitive if its transition monoid (semigroup) acting on the set of states has no non-trivial congruences.
It is synchronizing if it contains a constant map (transformation).
In analogy to synchronizing groups, we study the possibility of characterizing automata that are synchronizing if primitive.
We prove that the implication holds for several classes of automata.
In particular, we show it for automata whose every letter induce either a permutation or a semiconstant transformation (an idempotent with one point of contraction) unless all letters are of the first type.
We propose and discuss two conjectures about possible more general characterizations.
\end{abstract}
\section{Introduction}

We consider deterministic finite semiautomata (called shortly \emph{automata}) and the properties of their transition monoids.
An automaton is \emph{synchronizing} if it admits a \emph{reset} word, which is a word such that after reading it, the automaton is left in one known state, regardless of the initial state.
In the transition monoid of the automaton (or the semigroup acting on the set of states), this corresponds to the existence of a constant map, which is induced by a reset word.

The theory of synchronizing automata is most famous due to the \Cerny conjecture, which says that every synchronizing automaton with $n$ states admits a reset word of length at most $(n-1)^2$ \cite{C1964Poznamka}.
This longstanding open problem from~1969 motivated researchers to develop a vast number of results.
The currently best general upper bound is cubic in $n$ \cite{Pin1983OnTwoCombinatorialProblems,S2019ImprovementRecentUpperBound,S2018ImprovingTheUpperBound}.
Most of the research was collected in the recent survey \cite{V2022Survey}; see also the older ones \cite{KV2021Survey,V2008Survey}.

Here, we consider the possibilities of relating the primitivity of the transition monoid with its synchronizability.
Both these properties often appear together in the literature.

We define the primitivity in analogue to that in the context of permutation groups, which means that there is no non-trivial congruence on the set of states preserved by the action of letters.
This is often useful since we can then construct a smaller \emph{quotient} automaton, where a state represents a class in the original one.
So such congruences are used to derive results, in particular, upper bounds on the length of the shortest reset words in particular cases \cite{BFRS21SynchronizingStronglyConnectedPartialDFAs,GK2013AutomataRespectingIntervals,V2007WeaklyMonotonic}.
There are also some bounds stated for primitive synchronizing automata \cite{AR2016Semisimple} (primitive automata are called there \emph{simple}) and other results relating the primitivity and the synchronizability of automata \cite{H2021CompletelyPrimitiveStateSynchronizing,H2021SyncMaximalPermutationGroups}.

A very related problem concerning permutation groups was the subject of extensive research \cite{ABCRS2016PrimitiveGroups,ACS2017SynchronizationAndItsFriends,AS2006SynchronizingGroups}.
The problem was to characterize when a primitive permutation group, after adding one non-permutational transformation, results in a synchronizing semigroup.
The transformations of this property have been successfully characterized.
However, the whole study concerns only the case when the group is primitive, which is a strong condition from the automata point of view.
In many cases, the group contained in the transition monoid of an automaton is not only non-primitive but non-transitive or even trivial (just if the automaton has no permutational letters).
This question is naturally generalized to semigroups, where the non-permutational transformations also contribute to the primitiveness of the transition monoid.

Hence, our specific research question is the following:
Under what additional condition(s) every primitive automaton is synchronizing?

Additionally, if, in some cases, primitivity implies synchronizability, then sometimes we could relax necessary conditions, e.g., where an automaton needs to be both primitive and synchronizing \cite{R2015PrimitiveAndIrreducibleAutomata}, or we could obtain a synchronizing automaton when needed by ensuring its primitivity instead of other properties, e.g., we could think about auxiliary constructions that need to be synchronizing such as the \emph{induced automata} \cite{BFRS21SynchronizingStronglyConnectedPartialDFAs,BS2016AlgebraicSynchronizationCriterion}.

\subsection{Our contribution}

We propose criteria that presumably are sufficient to imply the synchronizability from the primitivity of an automaton.
We discuss two variants of the conjecture (weak and strong) referring to the shape of non-permutational transformations induced by the letters.
We show that they cannot be (much) relaxed and provide experimental support.

Based on the literature results, we show that the implication holds in several cases.
In particular, we prove that it holds for automata with permutational and semiconstant letters, which are a generalization of automata with simple idempotents \cite{H2022ConstrainedSynchronization,R2000EstimationSimpleIdempotents}.
This class is a restricted case covered by our conjecture in the strong variant.

\noindent\textbf{Note}:
The weak variant of our conjecture in a stronger form has been recently solved by Mikhail Volkov \cite{V2023SynchronizationOfPrimitiveAutomata}, together with several new results concerning our problem.

\section{Preliminaries}

An \emph{automaton} $\mathrsfs{A}$ is a triple $(Q, \Sigma, \delta)$, where $Q$ is a finite set of $n$ elements called \emph{states}, $\Sigma$ is a finite non-empty set of \emph{letters}, and $\delta \colon Q \times \Sigma$ is a totally-defined \emph{transition function}.
The transition function is naturally extended to words (finite sequences of elements from~$\Sigma$) in~$\Sigma^*$.
Given a word $w \in \Sigma^*$, let $\delta(w)$ denote its \emph{induced transformation}, which is a function (transformation, map) $\delta(w)\colon Q \to Q$ defined by $\delta(w)(q) = \delta(q,w)$.

For a set $T$ of maps $Q \to Q$, the transformation monoid $(Q, M)$ generated by $T$ and acting on $Q$ is denoted by $\langle T\rangle$, where $M$ is the set of all transformations that can be obtained from maps from $T$ by composition.
For an automaton $\mathrsfs{A} = (Q, \Sigma, \delta)$, its associated \emph{transition monoid} is $\langle \{\delta(a) \mid a \in \Sigma\}\rangle = (Q,M)$, where $M$ is the set of all maps $\delta(w)$ for every word $w \in \Sigma^*$.
In the literature, there are many names for a transformation monoid, e.g., transformation semigroup, operand \cite{CP1967TheAlgebraicTheoryOfSemigroups}, polygon \cite{ASSSS1991GeneralAlgebra}; and it is also an algebra with unary operations \cite{S1978CorrespondenceUnaryAlgebra}.

We use the right-hand side convention and denote by $(q)f$ the \emph{image} of the state $q \in Q$ under the action of $f\colon Q \to Q$.
For a subset of states $S \subseteq Q$, the \emph{image} of the map $f$ we denote by $(S) f = \{(q)f \mid q \in S\}$.
The cardinality $|(Q) f|$ is the \emph{rank} of $f$.
The \emph{deficiency} (or \emph{co-rank}) of $f$ is $n-|(Q) f|$.
Maps of rank $n$ (equivalently, of deficiency $0$) are permutations (bijections on $Q$).

A transformation monoid $(Q,M)$ is \emph{transitive} (or \emph{strongly connected}) if for every two states $s,t \in Q$, there is a transformation $f \in M$ such that $(s)f = t$.
If a transformation monoid contains a map of rank $1$, then the monoid is called \emph{synchronizing}.
A pair of different states $s,t \in Q$ are called \emph{compressible} if there is a transformation $f \in M$ such that $|(\{s,t\})f| = 1$; then $f$ \emph{compresses} the pair.
It is well known that a transformation monoid is synchronizing if and only if every pair of states is compressible \cite{C1964Poznamka}.

We transfer the above terminology from the transition monoid to the automaton and from transformations to words and letters that induce them.
Thus, a synchronizing automaton is one that admits a word of rank $1$.
Such a word is also a \emph{reset} word.

\subsection{Relations and Primitivity}

Let $(Q, M)$ be a transformation monoid.
Then the action of $M$ on $Q$ naturally continues on the square $Q \times Q$ by components:
\[ (s, t)f = ((s)f, (t)f), \text{ for each } f \in M, s,t \in Q .\]
For a binary relation $\rho \subseteq Q \times Q$, we put:
\[ (\rho)M = \{(s,t)f \mid (s,t) \in \rho, f \in M\} .\]
A relation $\rho$ is called \emph{invariant} for $(Q, M)$ if $(\rho)M = \rho$.
If $\rho$ is additionally an equivalence relation, then it is a \emph{congruence} of $(Q, M)$.
The equivalence classes of this congruence are called \emph{blocks}.

A relation $\rho$ is \emph{trivial} if it is the identity relation (\emph{diagonal}) $\IdentityRel_Q = \{(q,q) \mid q \in Q\}$ or the total relation (\emph{square}) $\TotalRel_Q = Q \times Q$.

\begin{definition}
A transformation monoid is \emph{primitive} if it does not have any non-trivial congruence.
An automaton is \emph{primitive} if its transition monoid is primitive.
\end{definition}

Letters of deficiency $0$, whose induced transformation is a permutation are called \emph{permutational}, and these transformations generate a permutation group acting on $Q$ and contained in the transition monoid.

\section{The Conjecture on Primitivity implying Synchronizability}

Every automaton with a letter of certain deficiency and a primitive group generated by the permutational letters is synchronising \cite{ABCRS2016PrimitiveGroups}.
This holds for deficiencies $1$, $2$, $n-3$, and $n-4$, and also for other deficiencies if the letter's map additionally has some properties.
Hence, it is natural to use deficiency as an additional condition for our question.

Our best guess for a general conjecture is the following:
\begin{conjecture}[Weak variant]\label{con:weak}
Every primitive automaton with permutational letters and letters of deficiency $1$ is synchronizing unless all letters are permutational.
\end{conjecture}

The class of these automata contains, in particular, \emph{almost-group} automata \cite{BN2020SynchronizingAlmostGroup}, which have only one non-permutational letter of deficiency $1$, and it is known that at random they are synchronizing with high probability.

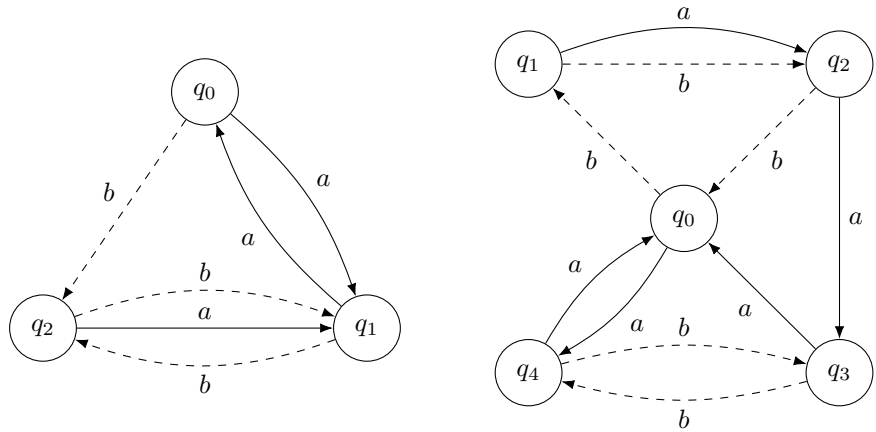
\begin{figure}[!htb]\centering\raisebox{0.1\totalheight}{
\begin{tikzpicture}[node distance=3cm,scale=1,every node/.style={transform shape},bend angle=30]
\node[state, inner sep=0pt] (q0) {$q_0$};
\node[state, inner sep=0pt] [below right=2.5cm and 1.5cm of q0] (q1) {$q_1$};
\node[state, inner sep=0pt] [below left=2.5cm and 1.5cm of q0] (q2) {$q_2$};
\draw (q0) edge[bend left=15] node[auto,midway]{$a$} (q1);
\draw (q1) edge[bend left=15] node[auto,midway]{$a$} (q0);
\draw (q2) edge node[auto,midway]{$a$} (q1);

\draw(q0) edge[dashed] node[auto,midway,swap]{$b$} (q2);
\draw(q1) edge[dashed,bend left=20] node[auto,midway]{$b$} (q2);
\draw(q2) edge[dashed,bend left=20] node[auto,midway]{$b$} (q1);
\end{tikzpicture}}
\hspace{1cm}
\begin{tikzpicture}[node distance=2cm,scale=1,every node/.style={transform shape},bend angle=30]
\node[state, inner sep=0pt] (q0) {$q_0$};
\node[state, inner sep=0pt] [above left=of q0] (q1) {$q_1$};
\node[state, inner sep=0pt] [above right=of q0] (q2) {$q_2$};
\node[state, inner sep=0pt] [below right=of q0] (q3) {$q_3$};
\node[state, inner sep=0pt] [below left=of q0] (q4) {$q_4$};

\draw (q1) edge[bend left=20] node[auto,midway]{$a$} (q2);
\draw (q2) edge node[auto,midway]{$a$} (q3);
\draw (q3) edge node[auto,midway]{$a$} (q0);
\draw (q0) edge[bend left=15] node[auto,midway]{$a$} (q4);
\draw (q4) edge[bend left=15] node[auto,midway]{$a$} (q0);

\draw (q1) edge[dashed] node[swap,midway]{$b$} (q2);
\draw (q2) edge[dashed] node[auto,midway]{$b$} (q0);
\draw (q0) edge[dashed] node[auto,midway]{$b$} (q1);
\draw (q3) edge[dashed,bend left=15] node[auto,midway]{$b$} (q4);
\draw (q4) edge[dashed,bend left=15] node[auto,midway]{$b$} (q3);
\end{tikzpicture}
\caption{Two non-primitive automata that are synchronizing. \textbf{Left:} Both letters have deficiency~$1$.
\textbf{Right:} Letter $a$ has deficiency~$1$ and letter $b$ is permutational.}
\label{fig:NonprimitiveSynchro}
\end{figure}

First, we note that the conjecture does not hold in the reverse direction.
There also exist almost-group automata that are synchronizing and non-primitive.

\begin{example}
Both automata from~\autoref{fig:NonprimitiveSynchro} are strongly connected, non-primitive, and synchronizing, though they have only permutational letters and letters of deficiency $1$.
The automaton on the right is almost-group.
\end{example}
\begin{proof}
The automata are synchronizing, since the words $ab$ and $a^2ba^3$ are reset respectively for the automaton on the left and on the right.
The automata are non-primitive, since the equivalence relation with classes $\{q_0,q_1\}$, $\{q_2\}$ is a non-trivial congruence for the automaton on the left and the equivalence relation with classes $\{q_0\}$, $\{q_1\}$, $\{q_2\}$, $\{q_3,q_4\}$ is a non-trivial congruence for the automaton on the right.
\end{proof}

\subsection{Known Cases}

By some results from the literature, we know that the conjecture holds in certain cases.

A \emph{sink} state $q_0 \in Q$ is such that $(q_0)M = \{q_0\}$.
An automaton is \emph{$0$-transitive} if it has a unique sink $q_0$ and $(q)M = Q$ for all $q \in Q \setminus \{q_0\}$;
in other words, there is a sink reachable from every state and the other states form one strongly connected component.

\begin{proposition}[{\cite[Proposition~5.1]{S2010TheoryOfTransformationMonoids}}]\label{prop:PrimitiveStronglyConnected}
For $n \ge 3$, a primitive automaton $(Q,\Sigma,\delta)$ is either strongly connected or $0$-transitive.
\end{proposition}
It follows that it would be enough to show the conjecture for strongly connected automata, as $0$-transitive automata are synchronized in $q_0$.

An automaton is called \emph{aperiodic} if its transition monoid does not have any non-trivial subgroups; in other words, in there is no transformation that acts cyclically on some subset of states of size $\ge 2$.

From the following statement, we get that a primitive aperiodic automaton is synchronizing:
\begin{proposition}[{\cite[rephrased Lemma~6]{T2007Aperiodic}}]
An aperiodic automaton $(Q,\Sigma,\delta)$ is synchronizing if and only if it has a state $q \in Q$ that is reachable from all the states, i.e., $q \in (p)M$ for all $p \in Q$.
\end{proposition}

Another case is automata with a prime number of states that also contains a full cyclic permutation in their transition monoid.
It follows by an old result of Pin \cite{Pin1978Circular}, stating that a \emph{circular} automaton, i.e., with a letter that induces one cycle on all the states, with a prime number of states is synchronizing whenever it has any non-permutational letter.

\begin{corollary}[{By~\cite{Pin1978Circular}}]
An automaton with a prime number of states $n$ that also contains a full cyclic permutation in its transition monoid is primitive and it is synchronizing if and only if it has any non-permutational letter.
\end{corollary}
\begin{proof}
The permutation group contained in the automaton's transition monoid is transitive, and it is known that every transitive group of a prime degree is primitive \cite{ACS2017SynchronizationAndItsFriends}.
The statement that the automaton is synchronizing if it contains any non-permutational letter has been proved in~\cite{Pin1978Circular}.
\end{proof}

\subsection{Strong Variant}

\begin{figure}[!htb]\centering\raisebox{0.05\totalheight}{
\begin{tikzpicture}[node distance=2cm,scale=1,every node/.style={transform shape},bend angle=30]
\node[state, inner sep=0pt] (q0) {$q_0$};
\node[state, inner sep=0pt] [right=of q0] (q1) {$q_1$};
\node[state, inner sep=0pt] [right=of q1] (q2) {$q_2$};
\node[state, inner sep=0pt] [below=of q2] (q3) {$q_3$};
\node[state, inner sep=0pt] [left=of q3] (q4) {$q_4$};

\draw (q0) edge[bend left=20] node[auto,midway]{$a$} (q1);
\draw (q1) edge[bend left=20] node[auto,midway]{$a$} (q0);
\draw (q4) edge node[auto,midway]{$a$} (q0);
\draw (q2) edge[bend left] node[auto,midway]{$a$} (q3);
\draw (q3) edge[loop, looseness=7,out=250,in=290] node[swap,midway]{$a$} (q3);

\draw(q0) edge[dashed,loop below,looseness=7,in=250,out=290] node[auto,midway]{$b$} (q0);
\draw(q1) edge[dashed] node[auto,midway]{$b$} (q2);
\draw(q2) edge[dashed] node[auto,midway,swap]{$b$} (q3);
\draw(q3) edge[dashed] node[auto,midway]{$b$} (q4);
\draw(q4) edge[dashed] node[auto,midway]{$b$} (q1);
\end{tikzpicture}}
\hspace{0.4cm}
\begin{tikzpicture}[node distance=2cm,scale=1,every node/.style={transform shape},bend angle=10]
\node[state, inner sep=0pt] (q0) {$q_0$};
\node[state, inner sep=0pt] [above left=of q0] (q1) {$q_1$};
\node[state, inner sep=0pt] [above right=of q0] (q2) {$q_2$};
\node[state, inner sep=0pt] [below right=of q0] (q3) {$q_3$};
\node[state, inner sep=0pt] [below left=of q0] (q4) {$q_4$};
\draw(q0) edge node[auto,midway]{$a$} (q3);
\draw(q2) edge node[auto,midway]{$a$} (q3);
\draw(q3) edge node[auto,midway]{$a$} (q4);
\draw(q4) edge node[auto,midway]{$a$} (q0);
\draw(q1) edge node[auto,midway]{$a$} (q0);

\draw(q1) edge[dashed,bend left] node[auto,midway]{$b$} (q4);
\draw(q4) edge[dashed,bend left] node[auto,midway]{$b$} (q1);
\draw(q0) edge[dashed,bend left] node[auto,midway]{$b$} (q2);
\draw(q2) edge[dashed,bend left] node[auto,midway]{$b$} (q0);
\draw(q3) edge[dashed,loop,out=-20,in=20,looseness=6] node[midway,swap]{$b$} (q3);

\draw(q1) edge[densely dotted,bend left] node[auto,midway]{$c$} (q2);
\draw(q2) edge[densely dotted,bend left] node[auto,midway]{$c$} (q1);
\draw(q0) edge[densely dotted,loop,out=250,in=290,looseness=6] node[midway,swap]{$c$} (q0);
\draw(q3) edge[densely dotted,loop,out=250,in=290,looseness=6] node[midway,swap]{$c$} (q3);
\draw(q4) edge[densely dotted,loop,out=250,in=290,looseness=6] node[midway,swap]{$c$} (q4);
\end{tikzpicture}
\caption{Two primitive automata that are not synchronizing. Letter $a$ has deficiency $2$, and letters $b$ and $c$ are permutational.}\label{fig:PrimitiveNonSynchro}
\end{figure}
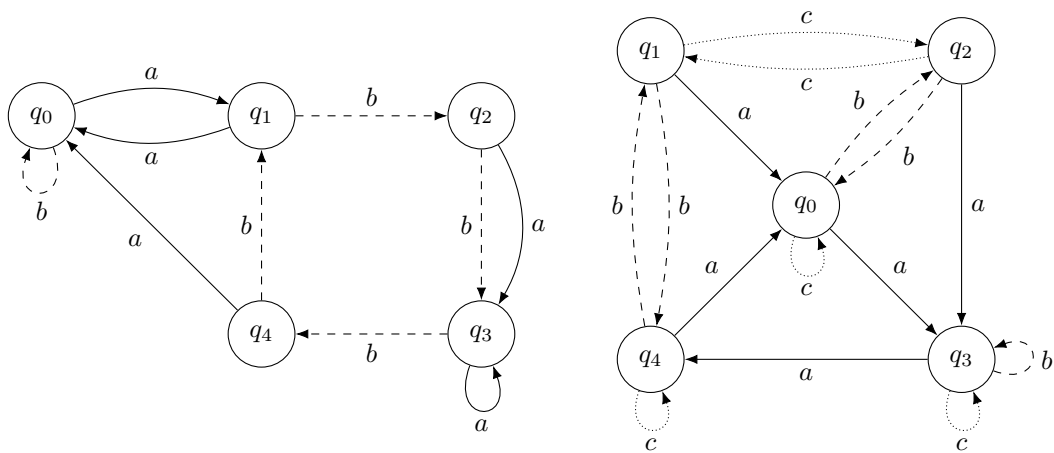

We note that restricting deficiency is essential.
Examples of primitive non-synchronizing automata can be easily found already if we admit letters of deficiency $2$.

\begin{example}\label{ex:PrimitiveNonSynchro}
Both automata from~\autoref{fig:PrimitiveNonSynchro} are primitive and non-synchronizing.
\end{example}
\begin{proof}
Consider the automaton on the left.
Observe that every pair of states can be mapped to $\{q_0,q_1\}$.
If there is a non-trivial congruence on the set of states, so some two states $q_i$, $q_j$ are in one block, then so $q_0$, $q_1$ are in one block.
But then also $q_0$ and each $q_i$ for $i=1,\ldots,4$ are in one block by the action of $b$, thus the congruence is trivial.
Also, the automaton is non-synchronizing, because the pairs $\{q_0,q_i\}$ for $i=1,\ldots,4$, $\{q_1,q_3\}$, and $\{q_2,q_4\}$ are not compressible.

Consider the automaton on the right.
Observe that every pair of states can be mapped to $\{q_0,q_3\}$, and this pair can be mapped to each $\{q_i,q_3\}$ for $i \in \{0,1,2,4\}$.
Thus, as before, the automaton is primitive.
Also, the pairs $\{q_i,q_3\}$ for $i \in \{0,1,2,4\}$, $\{q_0,q_4\}$, and $\{q_1,q_2\}$ are not compressible, thus the automaton is non-synchronizing.
\end{proof}

However, we can still slightly relax the condition and propose a stronger conjecture.
For this, we refer to the notion of components and cycles of a transformation, which are defined on the digraph $D(t)$ in which from every state $q$ there is exactly one outgoing arc $(q,(q)t)$.
Then, a \emph{component} of a transformation $t$ is a weakly connected component of $D(t)$, which always contains exactly one cycle, which can be a loop.
The strongly connected components of $D(t)$ are precisely its cycles.

\begin{conjecture}[Strong variant]\label{con:strong}
Every primitive automaton where each letter:
\begin{enumerate}
\item is permutational,
\item has deficiency $1$, or
\item its transformation has one component with a cycle of size $\le 2$ and the other components being cycles,
\end{enumerate}
is synchronizing, unless all letters are permutational.
\end{conjecture}

Both automata from~\autoref{fig:NonprimitiveSynchro} satisfy the conditions of~\autoref{con:strong} except that they are non-primitive; in particular, their letter $a$ meets condition~(3).
On the other hand, the examples from~\autoref{fig:PrimitiveNonSynchro} show that we cannot strengthen the conjecture more:
These automata do not satisfy the conditions due to the transformation of the letter $a$; in the left automaton, it has two components that are not cycles, and in the right automaton, the cycle in the unique component has size $3$.

\subsection{Experimental Support}

\begin{table}[!htb]\centering
$\begin{array}{|c|rrrrrrrr|}\hline
\text{Alphabet size}&\multicolumn{8}{c|}{\text{Number of states }n}\\
|\Sigma|               &    \le 4 &        5 & 6        & 7        & 8        & 9        & 10       & 11        \\\hline
2                      &\checkmark&\checkmark&\checkmark&\checkmark&\checkmark&\checkmark&\checkmark&\checkmark^*\\
3                      &\checkmark&\checkmark&\checkmark&\checkmark&          &          &          &           \\
4                      &\checkmark&\checkmark&\checkmark&          &          &          &          &           \\
5                      &\checkmark&\checkmark&          &          &          &          &          &           \\
6                      &\checkmark&\checkmark&          &          &          &          &          &           \\
\ge 7                  &\checkmark&          &          &          &          &          &          &           \\\hline
\end{array}$\\
\smallskip
$\checkmark$ -- Both \autoref{con:weak} and \autoref{con:strong} verified.\\
$\checkmark^*$ -- Only \autoref{con:weak} verified.
\smallskip
\caption{The range where our conjectures have been verified.}
\label{tab:experiments}
\end{table}

By the automata enumeration method of~\cite{KS2013GeneratingSmallAutomata,KKS2016ExperimentsWithSynchronizingAutomata}, we have verified both conjectures for automata with a small number of states and a small alphabet.
\autoref{tab:experiments} summarizes the cases.
The computation of the binary case with $n=11$ (only the weak variant) took about $81$ hours of one CPU core.

We also tried further strengthening of~\autoref{con:strong} and then it was easy to find many counterexamples already for $n \le 5$ (for instance, as in~\autoref{fig:PrimitiveNonSynchro}).

\section{Primitive Automata with Semiconstant Letters}

A map $f$ is \emph{semiconstant} if there is a subset $S \subseteq Q$ such that $(S)f = \{r\}$, for some state $r \in Q$, and $(q)f = q$ for every $q \in Q \setminus S$.
We denote it by $(S \to r)$.
The special cases of a semiconstant map are:
\begin{enumerate}
\item the \emph{identity} map, which can be denoted by $(\emptyset \to r)$ for any $r \in Q$;
\item \emph{constant} maps, which can be denoted by $(Q \to r)$ for $r \in Q$;
\item \emph{unitary} (or \emph{simple idempotent}) maps, denoted by $(\{q\} \to r)$, where $q \neq r$. We denote it simpler by $(q \to r)$.
\end{enumerate}
We transfer this terminology to letters that induce these transformations.

An automaton $\mathrsfs{A} = (Q,\Sigma,\delta)$ where every letter is permutational or semiconstant is a \emph{PSc-automaton}.
If additionally, every letter of the automaton is permutational or unitary, then it is a \emph{PU-automaton}.
In the literature, the latter was often called an \emph{automaton with simple idempotents} \cite{H2022ConstrainedSynchronization,R2000EstimationSimpleIdempotents}, whereas the names \emph{unitary} and \emph{semiconstant} appear in~\cite{BS2015LargeAperiodic}.
The later terminology is adopted as it offers and distinguish both types of letters that we consider here.

In this section, we prove that every primitive PSc-automaton is synchronizing.
We first show this for PU-automata and then generalize to the wider class.

\subsection{Closures of Relations}

Before we consider PU-automata, we state some basic properties of relations that will be used.

Let $(Q,M)$ be a transformation monoid.
It is easy to see that the union and the intersection of invariant relations are also invariant (for $(Q,M)$).
Since the identity relation $\IdentityRel_Q = \{(q,q) \mid q \in Q\}$ and the inverse relation $\rho^{-1} = \{(q,p) \mid (p,q) \in \rho\}$ of an invariant $\rho$ relation are also invariant, we get the following statement:

\begin{proposition}\label{prop:SymmetricClosure}
If $\rho$ is an invariant relation for $(Q,M)$, then its symmetric closure $\IdentityRel_Q \cup \rho \cup \rho^{-1}$ is also invariant for $(Q,M)$.
\end{proposition}

For a relation $\rho$, we assign the oriented graph $\Gamma(\rho) = (Q,\rho)$, called shortly the \emph{orgraph} of $\rho$.
For each pair of different states $(s,t) \in \rho$, there is an \emph{arc} (directed edge) in the orgraph from $s$ to $t$, and for $(s,s) \in \rho$, there is a loop.
When considering orgraphs, we use the usual notions from graph theory such as a path, a simple path, a route (walk), and a cycle.

Denote by $\rho^*$ the transitive closure of $\rho$ -- this is the reachability relation of the orgraph $\Gamma(\rho)$, i.e., $(s,t) \in \rho^*$ if and only if there is a (possibly empty) path from $s$ to $t$ in $\Gamma(\rho)$.

\begin{proposition}\label{prop:TransitiveClosure}
If $\rho$ is an invariant relation for $(Q,M)$, then also $\rho^*$ is invariant.
\end{proposition}
\begin{proof}
Let $(s,t) \in \rho^*$ and let $P=(s=s_1,\ldots,s_k=t)$ be a path in the orgraph $\Gamma(\rho)$ such that $(s_i,s_{i+1}) \in \rho$ for all $1 \le i < k$.
Consider a map $f \in M$ and the image $(P)f = ((s_1)f,\ldots,(s_k)f)$.
Then for all $1 \le i < k$, $((s_i)f,(s_{i+1})f) \in \rho$ since $\rho$ is invariant.
Thus, $(P)f$ is a route in $\Gamma(\rho)$, which implies that $((s)f,(t)f) \in \rho^*$.
\end{proof}

Denote by $\epsilon[\rho]$ the equivalence closure of $\rho$, i.e., the smallest equivalence relation containing $\rho$.
This can be obtained by taking the reflexive, symmetric, and transitive closures:
\[ \epsilon[\rho] = (\IdentityRel_Q\,\cup\,\rho\,\cup\,\rho^{-1})^* .\]
The equivalence classes of $\epsilon[\rho]$ are the weakly connected components of the orgraph $\Gamma(\rho)$.

From Propositions~\ref{prop:SymmetricClosure} and~\ref{prop:TransitiveClosure}, we get:
\begin{corollary}\label{cor:EquivalenceClosure}
If $\rho$ is an invariant relation of $(Q,M)$, then its equivalence closure $\epsilon[\rho]$ is a congruence of $(Q,M)$.
\end{corollary}

\begin{proposition}
\label{prop:EquivalenceClosureOfTransitive}
For a binary relation $\rho$, we have: $\epsilon[\rho] = \epsilon[\rho^*]$.
\end{proposition}
\begin{proof}
The direction $\subseteq$ follows from the trivial inclusion $\rho \subseteq \rho^*$, and the other direction $\supseteq$ follows from $\rho^* \subseteq \epsilon[\rho] = (\IdentityRel_Q \cup \rho \cup \rho^{-1})^*$ by definition.
\end{proof}

The following will be a useful property of relations:

\begin{definition}
A binary relation $\rho$ is \emph{cyclic} if its transitive closure is symmetric: $\rho^* = (\rho^*)^{-1}$.
\end{definition}

\begin{proposition}\label{prop:cyclic}
A binary relation $\rho$ is cyclic if and only if one of the following equivalent conditions are satisfied:
\begin{enumerate}
\item $\rho^{-1} \subseteq \rho^*$.
\item $\Gamma(\rho)$ is a union of (not necessarily disjoint) cycles.
\item Every weakly connected component of $\Gamma(\rho)$ is strongly connected.
\end{enumerate}
\end{proposition}
\begin{proof}
If $\rho$ is cyclic, then $\rho^{-1} \subseteq (\rho^*)^{-1} = \rho^*$ (1).

If for each arc $(s,t) \in \rho$, we have $(t,s) \in \rho^*$ (1), so there is a simple path $P$ from $s$ to $t$ in $\Gamma(\rho)$, then the path $(s,t) P$ is a cycle containing $(s,t)$ (2).

From~(2), every $s$ and $t$ from one weakly connected component are in a cycle, thus they are in the same strongly connected component (3).

From~(3), for every $(s,t) \in \rho^*$, there is also a path from $t$ to $s$, thus we have $(t,s) \in \rho^*$, and vice versa, thus the relation is cyclic.
\end{proof}

\begin{corollary}\label{cor:TransitiveClosureCyclic}
The transitive closure of a cyclic relation $\rho$ is an equivalence relation:
\[ \rho^* = \epsilon[\rho] .\]
\end{corollary}
\begin{proof}
By~\autoref{prop:EquivalenceClosureOfTransitive}, we have $\rho^* \subseteq \epsilon[\rho^*] = \epsilon[\rho]$.
Consider the other direction.
The identity relation $\IdentityRel_Q$ and $\rho$ are trivially contained in $\rho^*$.
From~\autoref{prop:cyclic}(1), we also have $\rho^{-1} \subseteq \rho^*$.
Thus, $\epsilon[\rho] \subseteq \rho^*$, and the equality follows.
\end{proof}

\subsection{Transition Semigroups of Automata with Permutational and Unitary Letters}

Let $\mathrsfs{A} = (Q,\Sigma,\delta)$ be a strongly connected PU-automaton, and let $P$ and $U$ be respectively the set of the permutations and the set of unitary transformations induced by its letters.
The transition monoid of this automaton is $(Q,\langle P \cup U\rangle)$.

Consider the submonoid $G = (Q,\langle P\rangle)$, which is a permutation group on $Q$.
Note that the reachability relation in $G$ is an equivalence relation, i.e., each weakly connected component is strongly connected, which is a class in this relation called an \emph{orbit}.
We denote by $(s)G$ the orbit of $G$ that contains $s \in Q$: $(s)G = \{(s)g \mid g \in G\}$.

The action of the permutation group $G$ continues on the square $Q \times Q$.
The (strongly) connected components of $G$ on $Q \times Q$ are called \emph{orbitals}.
Denote by $(s,t)G$ the orbital of $G$ that contains $(s,t)$, i.e., $(s,t)G = \{(s,t)g \mid g \in G\}$.
Note that each orbital is an invariant relation of $G$.
Two pairs of states are \emph{$G$-equivalent} if they belong to the same orbital of $G$.

We define the unitary relation $\delta(U)$ as follows:
\[ \delta(U) = \bigcup_{(s \to t) \in U} (s,t) .\]
The group closure of the idempotent relation is defined by the formula:
\begin{equation}
\pi = (\delta(U)) G = \bigcup_{(s,t) \in \delta(U)} (s,t)G . \label{eq:pi}
\end{equation}
Since $\pi$ is a union of orbitals of $G$, it is an invariant relation for $G$.

A pair of states $(s,t) \in \pi$ is called \emph{internal} if both states are in the same orbit of $G$.
The pair is \emph{external} if $s$ and $t$ are in different orbits of the group.
We split $\pi$ into disjoint $\pi_\mathrm{int}$ containing the internal pairs and $\pi_\mathrm{ext}$ containing the external pairs.

\begin{lemma}\label{lem:piInvariant}
The relations $\pi_\mathrm{int}$ and $\pi_\mathrm{ext}$ are invariant for $G$.
\end{lemma}
\begin{proof}
Since $\pi$ is invariant, for every $(s,t) \in \pi_\mathrm{int}$, we have $((s)g,(t)g) \in \pi$.
If $(s,t) \in \pi_\mathrm{int}$, thus $s$ and $t$ are in the same orbit of $G$, then we know that also $(s)g$ and $(t)g$ are in the same orbit, hence we have an internal arc: $((s)g,(t)g) \in \pi_\mathrm{int}$.
Similarly, if $(s,t) \in \pi_\mathrm{ext}$, thus $s$ and $t$ are in different orbits of $G$, then $(s)g$ and $(t)g$ are in different orbits, hence we have an external arc: $((s)g,(t)g) \in \pi_\mathrm{ext}$.
\end{proof}

The following statement in other terms has been used in the theory of permutation groups \cite{BBIT2021AlgebraicCombinatorics}.
\begin{lemma}\label{lem:piIntCyclic}
The relation $\pi_\mathrm{int}$ is cyclic.
\end{lemma}
\begin{proof}
Consider the orgraph $\Gamma(\pi_\mathrm{int}) = (Q,\pi_\mathrm{int})$ and an arc $(s,t) \in \pi_\mathrm{int}$.
Since $s$ and $t$ are in the same orbit of $G$, there is a permutation $g \in G$ such that $(s)g = t$.
Denote by $k \ge 1$ the minimum number such that $(t)g^k = s$, and consider the sequence of pairs $C = ((s,t)g^i\colon i \le i \le k)$, where $g^0$ is the identity map (the unit of $G$).
As $\pi_\mathrm{int}$ is invariant for $G$, all the arcs from $C$ are in $\Gamma(\pi_\mathrm{int})$ thus $C$ is a cycle in $\Gamma(\pi_\mathrm{int})$.
It follows that every arc from $\Gamma(\pi_\mathrm{int})$ is in a cycle, thus $\Gamma(\pi_\mathrm{int})$ is a union of cycles and by~\autoref{prop:cyclic}, $\pi_\mathrm{int}$ is cyclic.
\end{proof}

Consider the orgraph $\Gamma(\pi_\mathrm{ext})$ and let $(s,t) \in \pi_\mathrm{ext}$.
Then the pair $(s,t)$ is located between the orbits $(s)G$ and $(t)G$ of the group $G$.
Consider the orbital $(s,t)G$ of the pair $(s,t)$ and its orbital orgraph $\Gamma(s,t) = {((s)G \cup (t)(G), (s,t)G)}$.
It follows from formula~(\ref{eq:pi}) and Lemma~\ref{lem:piInvariant} that $\Gamma(s,t)$ is a subgraph of $\Gamma(\pi_\mathrm{ext})$.

\begin{lemma}\label{lem:bipartite}
For an external pair $(s,t) \in \pi_\mathrm{ext}$, the orbital orgraph $\Gamma(s,t)$ is a bipartite graph with parts $(s)G$ and $(t)G$ such that from each state in $(s)G$ goes out an arc to a state in the orbit $(t)G$.
\end{lemma}
\begin{proof}
For each state $s' \in (s)G$, we can find a permutation $g \in G$ such that $(s',t') = (s,t)g$ is the outgoing arc from $s'$ to a state $t' \in (t)G$.
\end{proof}

\begin{lemma}\label{lem:pathThroughOrbits}
For every $s,t \in Q$, there is a path in the orgraph $\Gamma(\pi_\mathrm{ext})$ from $s$ to a state in the orbit $(t)G$.
\end{lemma}
\begin{proof}
Since we assume that $\mathrsfs{A}$ is strongly connected, there is a path of transitions $e_1,\ldots,e_k$ from the state $s$ to a state in $(t)G$, i.e.,
for each $e_i = (s_i,t_i)$ there is either a permutational letter inducing $g \in P$ such that $(s_i)g = t_i$ or the unitary transition $(s_i \to t_i) \in U$. 
In the former case, $e_i$ is an internal arc from $\pi_\mathrm{int}$, and in the latter case, it can be either an internal or an external arc from $\pi$.

We remove all internal arcs, obtaining a sequence of external arcs $e_{i_1}, \ldots, e_{i_m}$ only, which is a subsequence of $e_1,\ldots,e_k$.
If this is the empty sequence ($m=0$), then $s \in (t)G$, so the lemma holds.
Otherwise, for the first arc $e_{i_1} = (s_{i_1},t_{i_1})$, we have $s_{i_1}$ in the orbit $(s)G$ and $t_{i_1}$ in another orbit.
By~\autoref{lem:bipartite} applied for $e_{i_1}$, we can find a pair $(s,s_1) \in \pi_\mathrm{ext}$ that is $G$-equivalent to $e_{i_1}$, thus where $s_1$ is a state from $(t_{i_1}) G$.
There may be several such arcs, and we can choose any.
Then for $e_{i_2}$, we choose an arc $(s_1,s_2) \in \pi_\mathrm{ext}$ in the same way, and so on till $e_{i_m}$.
As the result, we obtain a route in the orgraph $\Gamma(\pi_\mathrm{ext})$ that starts from $s$, follows the same orbits as the path $e_1,\ldots,e_k$, and ends in a state from $(t)G$.
\end{proof}

\begin{lemma}\label{lem:piExtCyclic}
The relation $\pi_\mathrm{ext}$ is cyclic.
\end{lemma}
\begin{proof}
Consider the orgraph $\Gamma(\pi_\mathrm{ext}) = (Q,\pi_\mathrm{ext})$ and an arc $(s_1,t_1) \in \pi_\mathrm{ext}$.
We are going to show that $(s_1,t_1) \in \pi_\mathrm{ext}$ is in some cycle in $\Gamma(\pi_\mathrm{ext})$, thus the lemma follows by~\autoref{prop:cyclic}.

By~\autoref{lem:pathThroughOrbits}, there is a path $P_1$ in the orgraph $\Gamma(\pi_\mathrm{ext})$ from $t_1$ to some state $s_2 \in (s_1)G$.
If $s_1 = s_2$, then we have found the cycle $(s_1,t_1) P_1$, which contains $(s_1,t_1)$.
Otherwise, we choose a state $t_2 \in (t_1)G$ such that the arc $(s_2,t_2) \in \pi_\mathrm{ext}$ is $G$-equivalent to $(s_1,t_1)$, and we again choose a path $P_2$ from $t_2$ to some state $s_3 \in (s_1)G$.
We continue this process, which stops when the found path $P_k$ ends in a state $s_{k+1} \in (s_1)G$ that has been encountered before as an $s_i$ for some $i \le k$.
As the size of $(s_1)G$ is finite, this finally happens.

Now, $(s_i,t_i) P_i \ldots (s_k,t_k) P_k$ is a closed route containing $(s_i,t_i) \in \pi_\mathrm{ext}$.
Let $X$ be the path separated from the route $P_i \ldots (s_k,t_k) P_k$ by removing all cycles so that every state occurs at most once.
Then $C = (s_i,t_i) X$ is a cycle containing the external arc $(s_i,t_i)$, which is $G$-equivalent to $(s_1,t_1)$.
Let $g \in G$ be such that $(s_1,t_1) = (s_i,t_i)g$.
Then the cycle $C' = (C)g$ is contained in $\Gamma(\pi_\mathrm{ext})$ (as $\pi_\mathrm{ext}$ is invariant for $G$) and contains $(s_1,t_1)$.
\end{proof}

\begin{corollary}\label{cor:piCyclic}
The relation $\pi = \pi_\mathrm{int} \cup \pi_\mathrm{ext}$ is cyclic.
\end{corollary}

We recall the known characterization of synchronizing automata with permutational and unitary letters.
In the literature, the orgraph $\Gamma(\pi)$ is often called the Rystsov graph of an automaton \cite{BV2018CharacterizationOfCompletelyReachable,CV2022BinaryCompletelyReachable,H2022ConstrainedSynchronization}.

\begin{theorem}[{\cite[Corollary~2 and Theorem~1]{R2000EstimationSimpleIdempotents}}]\label{thm:PU-Synchronizing}
A strongly connected PU-automaton is synchronizing if and only if the orgraph $\Gamma(\pi)$ is strongly connected.
\end{theorem}

We have now all ingredients for the main result of this section.

\begin{theorem}\label{thm:U-Synchronizing}
A primitive PU-automaton with at least one unitary letter is synchronizing.
\end{theorem}
\begin{proof}
By~\autoref{prop:PrimitiveStronglyConnected}, it is enough to consider a strongly connected primitive $U$-automaton.
Suppose that such an automaton is not synchronizing, thus by~\autoref{thm:PU-Synchronizing}, the orgraph $\Gamma(\pi)$ is not strongly connected.
Then, for the transitive closure $\pi^*$ of $\pi$, we obtain the strict inclusion: $\pi^* \subsetneq Q \times Q$.
Since the automaton has at least one unitary letter, the relation $\pi$ is non-empty and thus $\IdentityRel \subsetneq \pi^*$.
Since $\pi$ is cyclic (\autoref{cor:piCyclic}), from~\autoref{cor:TransitiveClosureCyclic}, we get $\pi^* = \epsilon[\pi]$ so it is an equivalence relation.
As $\pi$ is invariant for $G$ (\autoref{lem:piInvariant}), from~\autoref{cor:EquivalenceClosure}, the equivalence closure $\epsilon[\pi]$ is a congruence of the automaton's transition monoid $(Q,\langle P \cup U\rangle)$.
Thus, $\epsilon[\pi]$ is a non-trivial congruence of $(Q,\langle P \cup U\rangle)$, so the automaton is non-primitive.
\end{proof}

\subsection{Generalization to Semiconstant Letters}

The generalization follows by a decomposition of semiconstant transformations into unitary transformations.
If $\mathrsfs{A}$ is an PSc-automaton, then by $\mathrsfs{A}^{\mathrm{Sc}\to\mathrm{U}}$ we denote the PU-automaton obtained from $\mathrsfs{A}$ in the following way: every letter that induces a semiconstant transformation $(\{s_1,\ldots,s_k\} \to r)$, where the states $s_i$ are pairwise distinct and different from $r$, is replaced with $k$ fresh letters inducing the unitary transformations $(s_1 \to r), \ldots, (s_k \to r)$, respectively.

\begin{lemma}\label{lem:ScToU}
Let $\mathrsfs{A}$ be a PSc-automaton. Then:
\begin{enumerate}
\item If $\mathrsfs{A}$ is primitive then also $\mathrsfs{A}^{\mathrm{Sc}{\to}\mathrm{U}}$ is primitive.
\item $\mathrsfs{A}$ is synchronizing if and only if $\mathrsfs{A}^{\mathrm{Sc}{\to}\mathrm{U}}$ is synchronizing.
\end{enumerate}
\end{lemma}
\begin{proof}
A semiconstant transformation $(\{s_1,\ldots,s_k\} \to r)$ can be generated as the composition of the $k$ unitary transformations that replace it:
\[ (s_1 \to r)\cdot\ldots\cdot(s_k \to r) = (\{s_1,\ldots,s_k\} \to r) .\]
Hence, the transition monoid of $\mathrsfs{A}$ is a submonoid of $\mathrsfs{A}^{\mathrm{Sc}{\to}\mathrm{U}}$.
This implies~(1) and the $\Rightarrow$ direction of~(2).

For the $\Leftarrow$ direction of~(2), denote $\mathrsfs{A} = (Q, \Sigma_\mathrsfs{A}, \delta_\mathrsfs{A})$ and $\mathrsfs{B} = (Q, \Sigma_\mathrsfs{B}, \delta_\mathrsfs{B}) = \mathrsfs{A}^{\mathrm{Sc}{\to}\mathrm{U}}$.
We show by induction on $k$ that every pair of states $\{p,q\}$ which is compressible in $\mathrsfs{B}$ with a word of length $k$ is compressible in $\mathrsfs{A}$.
For $k=0$ it is trivial.
Assume the statement for $k$ and consider a pair $\{p,q\}$ which is compressible in $\mathrsfs{B}$ with a word $w = a u$ of length $k+1$, where $a \in \Sigma_\mathrsfs{B}$ and $u \in \Sigma^*_\mathrsfs{B}$.

If $a$ is permutational, then let $\tau(a) = a$, and the following equation holds trivially:
\begin{equation}
\big\{\delta_\mathrsfs{B}(p,a),\delta_\mathrsfs{B}(q,a)\big\} = \big\{\delta_\mathrsfs{A}(p,\tau(a)),\delta_\mathrsfs{A}(q,\tau(a))\big\} \label{eq:BtoA}
\end{equation}
Then the statement follows from the induction hypothesis since $u$ of length $k$ compresses the pair from~(\ref{eq:BtoA}) in~$\mathrsfs{B}$.

Otherwise, $a$ induces a unitary transformation $(s \to r)$.
If $s \notin \{p,q\}$, then $\delta_\mathrsfs{B}(\{p,q\},a) = \{p,q\}$, so again $u$ of length $k$ compresses the pair.
Without loss of generality, suppose $s = p$.
Then let $\tau(a)$ be a semiconstant letter that induces $(S \to r)$, where $s=p \in S$.
If $q \in S \cup \{r\}$, then $\tau(a)$ compresses $\{p,q\}$, so we are done.
Otherwise, $q \notin S \cup \{r\}$, thus (\ref{eq:BtoA}) holds again.

It follows that if $\mathrsfs{B}$ is synchronizing, thus all pairs of states are compressible, then also they are in $\mathrsfs{A}$, so $\mathrsfs{A}$ is synchronizing.
\end{proof}

\begin{theorem}
A primitive PSc-automaton is synchronizing, unless all its letters are permutational.
\end{theorem}
\begin{proof}
If $\mathrsfs{A}$ is a primitive PSc-automaton, then by~\autoref{lem:ScToU}(1), $\mathrsfs{A}^{\mathrm{Sc}{\to}\mathrm{U}}$ is primitive, so by~\autoref{thm:U-Synchronizing} $\mathrsfs{A}^{\mathrm{Sc}{\to}\mathrm{U}}$ is synchronizing.
Then by~\autoref{lem:ScToU}(2), $\mathrsfs{A}$ is also synchronizing.
\end{proof}

\bibliography{bibliography}
\end{document}